\documentclass{article}

\usepackage{PRIMEarxiv}

\usepackage[utf8]{inputenc}
\usepackage[T1]{fontenc}
\usepackage{hyperref}
\usepackage{url}
\usepackage{booktabs}
\usepackage{amsfonts}
\usepackage{nicefrac}
\usepackage{microtype}
\usepackage{graphicx}

\usepackage[numbers]{natbib}
\usepackage{hyperref}

\usepackage{booktabs}
\usepackage{tabularx}
\usepackage{longtable}
\usepackage{multirow}
\usepackage{array}
\usepackage{makecell}
\usepackage{rotating}

\newcolumntype{L}[1]{>{\raggedright\arraybackslash}p{#1}}
\newcolumntype{C}[1]{>{\centering\arraybackslash}p{#1}}

\usepackage{enumitem}
\usepackage{siunitx}
\usepackage{comment}
\usepackage{setspace}

\usepackage{graphicx}
\usepackage{tikz}
\usetikzlibrary{
    positioning,
    shapes.geometric,
    shapes.misc,
    arrows.meta,
    calc,
    mindmap,
    shadows,
    fit,
    backgrounds,
    patterns
}

\usepackage{pgfplots}
\pgfplotsset{compat=1.18}

\usepackage{algorithm}
\usepackage{algpseudocode}
\algrenewcommand\algorithmicrequire{\textbf{Input:}}
\algrenewcommand\algorithmicensure{\textbf{Output:}}

\definecolor{amber}{rgb}{0.75, 0.50, 0.00}
\definecolor{teal} {rgb}{0.00, 0.50, 0.50}

\definecolor{reviewerblue} {RGB}{0,70,140}
\definecolor{commentgray}  {RGB}{240,240,245}
\definecolor{responsegreen}{RGB}{0,100,60}

\definecolor{grpFormal}{RGB}{180,160,220}   % violet (light)
\definecolor{grpPlan}  {RGB}{140,210,185}   % teal   (light)
\definecolor{grpAgent} {RGB}{250,185,130}   % orange (light)
\definecolor{grpRun}   {RGB}{250,215,100}   % amber  (light)
\definecolor{grpBench} {RGB}{195,195,185}   % gray   (light)

\definecolor{ldblue}  {HTML}{1565C0}
\definecolor{trgreen} {HTML}{1B5E20}
\definecolor{cpurp}   {HTML}{4A148C}
\definecolor{esorg}   {HTML}{BF360C}
\definecolor{safegrn} {HTML}{2E7D32}
\definecolor{unsafrd} {HTML}{B71C1C}
\definecolor{cexbrn}  {HTML}{4E342E}
\definecolor{diagpnk} {HTML}{880E4F}
\definecolor{panelbg} {HTML}{F5F5F5}

\definecolor{figviolet}{RGB}{180,160,220}
\definecolor{figamber} {RGB}{250,185,100}
\definecolor{figteal}  {RGB}{180,225,210}
\definecolor{figblue}  {RGB}{180,210,240}
\definecolor{figgray}  {RGB}{220,220,215}

\definecolor{codebg}    {HTML}{F8F8F8}
\definecolor{codeframe} {HTML}{DDDDDD}
\definecolor{codegreen} {HTML}{2E7D32}
\definecolor{codeblue}  {HTML}{1565C0}
\definecolor{codegray}  {HTML}{757575}
\definecolor{codepurple}{HTML}{7B1FA2}

\usepackage{listings}
\lstdefinelanguage{XML}{
  morestring       = [b]",
  morestring       = [s]{>}{<},
  morecomment      = [s]{<!--}{-->},
  stringstyle      = \color{codepurple},
  identifierstyle  = \color{codeblue},
  keywordstyle     = \color{codegreen},
  morekeywords     = {localId,refLocalId,variable,negated,storage,
                      connectionPointIn,connection,leftPowerRail,
                      rightPowerRail,contact,coil,LD,rung}
}

\usepackage[acronym, nogroupskip, nonumberlist, nopostdot]{glossaries}
\glsdisablehyper

\newcommand{\doi}[1]{\url{https://doi.org/#1}}

\usepackage{siunitx}
\usepackage{amsmath}
\usepackage{amssymb}
\usepackage{amsthm}

\newtheoremstyle{nonitalic}
  {}{}                      % espaço antes, espaço depois
  {\normalfont}              % estilo do corpo — normal, não itálico
  {}                         % indentação
  {\bfseries}{.}              % estilo do título, pontuação
  { }                         % espaço após o título
  {}                          % especificação do cabeçalho (vazio = padrão)

\theoremstyle{nonitalic}
\newtheorem{definition}{Definition}[section]

\theoremstyle{nonitalic}
\newtheorem{theorem}{Theorem}[section]

\usepackage{xcolor}
\definecolor{codehl}{gray}{0.92}

\usepackage{soul}
\sethlcolor{codehl}
\newcommand{\code}[1]{\hl{\texttt{#1}}}
\newacronym{dfs}{DFS}{Depth-First Search}

\newacronym{aadl}{AADL}{Architecture Analysis and Design Language}
\newacronym{acm}{ACM}{Association for Computing Machinery}
\newacronym{ai}{AI}{Artificial Intelligence}
\newacronym{ait}{AIT}{Timing Analyzer}
\newacronym{ansi-c}{ANSI-C}{American National Standards Institute C}
\newacronym{api}{API}{Application Programming Interface}
\newacronym{arinc}{ARINC}{Aeronautical Radio, Incorporated}
\newacronym{ase}{ASE}{Automated Software Engineering}
\newacronym{asic}{ASIC}{Application-Specific Integrated Circuit}
\newacronym{ast}{AST}{Abstract Syntax Tree}
\newacronym{autosar}{AUTOSAR}{AUTomotive Open System ARchitecture}
\newacronym{bdd}{BDD}{Binary Decision Diagrams}
\newacronym{bmc}{BMC}{Bounded Model Checking}
\newacronym{capes}{CAPES}{Brazilian Federal Agency for Support and Evaluation of Graduate Education}
\newacronym{cbmc}{CBMC}{Bounded Model Checking for ANSI-C Programs}
\newacronym{cca}{CCA}{Confidential Compute Architecture}
\newacronym{cegar}{CEGAR}{Counterexample-Guided Abstraction Refinement}
\newacronym{cern}{CERN}{Conseil Européen pour la Recherche Nucléaire}
\newacronym{ceteli}{CETELI}{Research and Development Center in Electronic Information Technology}
\newacronym{cfg}{CFG}{Control Flow Graph}
\newacronym{chain-of-thought}{CoT}{Chain-of-Thought}
\newacronym{chc}{CHC}{Constrained Horn Clause}
\newacronym{cheri}{CHERI}{Capability Hardware Enhanced RISC Instructions}
\newacronym{cicd}{CI/CD}{Continuous Integration and Continuous Deployment}
\newacronym{cisq}{CISQ}{Consortium for Information and Software Quality}
\newacronym{cli}{CLI}{Command-Line Interface}
\newacronym{cordis}{CORDIS}{Community Research and Development Information Service}
\newacronym{cpu}{CPU}{Central Processing Unit}
\newacronym{cuda}{CUDA}{Compute Unified Device Architecture}
\newacronym{cve}{CVE}{Common Vulnerability and Exposure}
\newacronym{darpa}{DARPA}{Defense Advanced Research Projects Agency}
\newacronym{dblp}{DBLP}{Digital Bibliography \& Library Project}
\newacronym{defi}{DeFi}{Decentralized Finance}
\newacronym{do-178c}{DO-178C}{Software Considerations in Airborne Systems and Equipment Certification}
\newacronym{do-254}{DO-254}{Design Assurance for Airborne Electronic Hardware}
\newacronym{do-330}{DO-330}{Software Tool Qualification Considerations}
\newacronym{dpll}{DPLL}{Davis-Putnam-Logemann-Loveland}
\newacronym{dsl}{DSL}{Domain-Specific Language}
\newacronym{ecs}{ECS}{Ethereum Consensus Specification}
\newacronym{eda}{EDA}{Electronic Design Automation}
\newacronym{epsrc}{EPSRC}{Engineering and Physical Sciences Research Council}
\newacronym{esbmc}{ESBMC}{Efficient SMT-based Context-Bounded Model Checker}
\newacronym{evm}{EVM}{Ethereum Virtual Machine}
\newacronym{fase}{FASE}{Fundamental Approaches to Software Engineering}
\newacronym{fbd}{FBD}{Functional Block Diagram}
\newacronym{ffi}{FFI}{Foreign Function Interface}
\newacronym{fpga}{FPGA}{Field-Programmable Gate Array}
\newacronym{gdp}{GDP}{Gross Domestic Product}
\newacronym{gnat}{GNAT}{GNAT Ada compiler}
\newacronym{gpt-4}{GPT-4}{Generative Pre-trained Transformer 4}
\newacronym{gpu}{GPU}{graphical processing unit}
\newacronym{gsn}{GSN}{Goal Structuring Notation}
\newacronym{ic3}{IC3}{Incremental Construction of Inductive Clauses for Indubitable Correctness}
\newacronym{icse}{ICSE}{International Conference on Software Engineering}
\newacronym{ide}{IDE}{Integrated Development Environment}
\newacronym{iec}{IEC}{International Electrotechnical Commission}
\newacronym{ieee}{IEEE}{Institute of Electrical and Electronics Engineers}
\newacronym{ikos}{IKOS}{Inference Kernel for Open Static Analyzers}
\newacronym{iot}{IoT}{Internet of Things}
\newacronym{ir}{IR}{Intermediate Representation}
\newacronym{iso}{ISO}{International Organization for Standardization}
\newacronym{issta}{ISSTA}{International Symposium on Software Testing and Analysis}
\newacronym{jvm}{JVM}{Java Virtual Machine}
\newacronym{ld}{LD}{Ladder Diagram}
\newacronym{llb}{LLB}{Ladder Logic Bombs}
\newacronym{llm}{LLM}{Large Language Model}
\newacronym{mcas}{MCAS}{Maneuvering Characteristics Augmentation System}
\newacronym{mcdc}{MCDC}{Modified Condition/Decision Coverage}
\newacronym{mir}{MIR}{Mid-level Intermediate Representation}
\newacronym{musl}{MUSL}{musl C standard library}
\newacronym{nasa}{NASA}{National Aeronautics and Space Administration}
\newacronym{nhtsa}{NHTSA}{National Highway Traffic Safety Administration}
\newacronym{nist}{NIST}{National Institute of Standards and Technology}
\newacronym{pdr}{PDR}{Property Directed Reachability}
\newacronym{plc}{PLC}{Programmable Logic Controller}
\newacronym{por}{POR}{Partial Order Reduction}
\newacronym{ppgee}{PPGEE}{Graduate Program in Electrical Engineering}
\newacronym{propesp}{PROPESP}{Office of the Vice-Rector for Research and Graduate Studies}
\newacronym{raii}{RAII}{Resource Acquisition Is Initialization}
\newacronym{rmm}{RMM}{Realm Management Monitor}
\newacronym{rtl}{RTL}{Register Transfer Level}
\newacronym{rtos}{RTOS}{Real-Time Operating System}
\newacronym{sac}{SAC}{Symposium on Applied Computing}
\newacronym{sat}{SAT}{Boolean Satisfiability}
\newacronym{sbseg}{SBSeg}{Brazilian Symposium on Cybersecurity}
\newacronym{slr}{SLR}{Systematic Literature Review}
\newacronym{sme}{SME}{Small and Medium-sized Enterprise}
\newacronym{smt}{SMT}{Satisfiability Modulo Theories}
\newacronym{smtlib2}{SMT-LIB}{Satisfiability Modulo Theories Library}
\newacronym{soc}{SoC}{System on Chip}
\newacronym{sos}{SOS}{Structural Operational Semantics}
\newacronym{sri}{SRI}{Stanford Research Institute}
\newacronym{ssa}{SSA}{Static Single Assignment}
\newacronym{ssvlab}{SSVLab}{Systems and Software Verification LAB}
\newacronym{st}{ST}{Structured Text}
\newacronym{sttt}{STTT}{Software Tools for Technology Transfer}
\newacronym{svcomp}{SV-COMP}{Competition on Software Verification}
\newacronym{tacas}{TACAS}{Tools and Algorithms for the Construction and Analysis of Systems}
\newacronym{tdd}{TDD}{Test-Driven Development}
\newacronym{test-comp}{Test-Comp}{Competition on Software Testing}
\newacronym{tse}{TSE}{Transactions on Software Engineering}
\newacronym{tvl}{TVL}{Total Value lLocked}
\newacronym{ufam}{UFAM}{Federal University of Amazonas}
\newacronym{ukri}{UKRI}{UK Research and Innovation}
\newacronym{vhdl}{VHDL}{VHSIC Hardware Description Language}
\newacronym{wcet}{WCET}{Worst-Case Execution Time}
\newacronym{il}{IL}{Instruction List}
\newacronym{sfc}{SFC}{Sequential Function Chart}
\newacronym{ics}{ICS}{Industrial Control Systems}

\newacronym{ltl}{LTL}{Linear Temporal Logic}
\newacronym{ctl}{CTL}{Computation Tree Logic}

\newacronym{scl}{SCL}{Structured Control Language}
\newacronym{stl}{STL}{Statement List}

\newacronym{dag}{DAG}{Directed Acyclic Graph}

\usepackage{xcolor}
\title{ESBMC-GraphPLC: Formal Verification of Graphical PLCopen XML Ladder Diagram Programs Using SMT-Based Model Checking}

\author{
  Pierre Dantas \\
  Computer Science, The University of Manchester \\
  Manchester, UK \\
  \texttt{pierre.dantas@manchester.ac.uk} \\
  \And
  Lucas Cordeiro \\
  Computer Science, The University of Manchester \\
  Manchester, UK \\
  \texttt{lucas.cordeiro@manchester.ac.uk} \\
  \And
  Waldir Junior \\
  Electrical Engineering, Federal University of Amazonas (UFAM) \\
  Manaus, AM, Brazil \\
  \texttt{waldirjr@ufam.edu.br} \\
}

\begin{document}
\maketitle
\glsresetall

\begin{abstract}
PLCopen XML (\code{tc6\_0201}) defines two encoding formats for \gls{iec}~61131-3 Ladder Diagram programs: a \emph{textual} format using explicit \code{<rung>} elements, and a \emph{graphical} format that encodes rung logic as a directed graph of \code{localId}/\code{refLocalId} connections. ESBMC-PLC supported the textual format and identified the graphical format as a known gap: files exported by CONTROLLINO, Beremiz, and OpenPLC Editor were parsed but produced an empty GOTO intermediate representation, causing verification to succeed \emph{vacuously} -- all properties held trivially because no rung logic was generated and all variables remained zero-initialized. This paper presents \textbf{ESBMC-GraphPLC}, which closes this gap by introducing a \gls{dfs}-based graphical \gls{ld} resolver. The resolver traverses the \code{localId}/\code{refLocalId} connection graph from \code{leftPowerRail} to each coil, extracts rung paths as Boolean conjunctions of contacts, and applies a three-tier I/O inference scheme (\%-address-based, then heuristic contact/coil-only analysis). A critical implementation insight -- ordering coils by their \code{rightPowerRail} \code{connectionPointIn} sequence -- ensures that SET coils are processed before RESET coils, matching IEC~61131-3 scan-cycle semantics. The graphical-to-IR conversion requires no changes to the ESBMC backend: the GOTO IR encoder, \textit{k}-induction engine, Z3 SMT solver, and property encoder are unchanged. ESBMC-GraphPLC is validated on 3~graphical \gls{ld} programs from CONTROLLINO/OpenPLC~Editor exports. All 3~programs produce a full GOTO \gls{ir} with nondeterministic inputs and rung logic, whereas the previous work produced an empty \gls{ir}. All 3~programs verify SAFE at \textit{k}=2 in under \SI{70}{\milli\second}. The 11~textual \gls{ld} benchmarks from ESBMC-PLC that exist in genuine textual format are fully preserved: 11/11 produce identical results with zero regressions. Benchmark curation also surfaced two Beremiz example programs with no Ladder Diagram content or with timer-block semantics not yet preserved by the resolver; both are reported as discovered limitations rather than counted as successes. The complete artefact is archived at Zenodo~\cite{DantasCordeiro2026graphical} (\doi{10.5281/zenodo.20699856}).
\end{abstract}


\glsresetall

%=============================================================================
\section{Introduction}
\label{sec:intro}
%=============================================================================

\glspl{plc} govern safety-critical processes in nuclear power stations, water treatment plants, chemical refineries, and railway signalling systems. The dominant \gls{plc} programming notation -- \gls{ld}, standardised in \mbox{\gls{iec}~61131-3} -- is exchanged between development tools in PLCopen XML format (\code{tc6\_0201}). This standard defines \emph{two} distinct encoding formats for the same logical content:

\textbf{Textual \gls{ld}} uses explicit \code{<rung>} elements in which contacts, coils, and function blocks are direct children in evaluation order, as shown in Listing~\ref{lst:textual-xml}.

\begin{lstlisting}[language=XML, caption={Textual PLCopen XML rung}, label={lst:textual-xml}]
<rung>
  <contact variable="Pool_Low"/>
  <contact variable="Tank_High" negated="true"/>
  <coil variable="Water_Pump"/>
</rung>
\end{lstlisting}

\textbf{Graphical \gls{ld}} encodes rung topology as a directed graph. Each element carries a \code{localId}, and connectivity is specified via \code{refLocalId} references nested inside \code{connectionPointIn} children, as shown in Listing~\ref{lst:graphical-xml}.

\begin{lstlisting}[language=XML, caption={Graphical PLCopen XML (\code{tc6\_0201}) -- same logic as Listing~\ref{lst:textual-xml}}, label={lst:graphical-xml}]
<contact localId="3">
  <connectionPointIn>
    <connection refLocalId="9"/>
  </connectionPointIn>
  <variable>Pool_Low_Level_Sensor</variable>
</contact>
<coil localId="4" storage="set">
  <connectionPointIn>
    <connection refLocalId="6"/>
    <connection refLocalId="12"/>
  </connectionPointIn>
  <variable>Water_Pump</variable>
</coil>
\end{lstlisting}

Graphical \gls{ld} is the native export format of the CONTROLLINO/OpenPLC Editor, Beremiz \gls{ide}, and most commercial \gls{plc} development environments. The textual format is a minority encoding used primarily in legacy files and hand-authored XML.

\subsection{The Gap: Vacuous Verification of Graphical \gls{ld}}

ESBMC-PLC~\cite{DantasCordeiro2026artefact} introduced the first open-source formal verifier with native support for textual PLCopen XML \gls{ld}. The tool's parser accepted graphical \code{tc6\_0201} files without error. Still, the \gls{ld}-to-GOTO-\gls{ir} converter found no \code{<rung>} elements, emitted no rung assignment instructions, and left all variables at their zero-initialized defaults. The resulting GOTO program is shown in Listing~\ref{lst:empty-ir}.

\begin{lstlisting}[caption={GOTO IR emitted by ESBMC-PLC for a graphical LD program (schematic)}, label={lst:empty-ir}]
/* All variables zero-initialised; no rung logic emitted */
static bool Water_Pump = false;
while (1) {
    /* EMPTY -- no rung assignments */
    assert(!Water_Pump || Pool_Low); /* trivially true: both false */
}
\end{lstlisting}

Every property asserted over zero-initialized variables trivially holds: \code{ESBMC} found no violation, returned \code{SAFE}, and \textit{k}-induction converged at $k=1$ with no inductive content. The proofs were \emph{vacuous} -- sound in the degenerate sense that an empty program satisfies any invariant, but providing no guarantee about the actual \gls{plc} logic.

This gap was documented in ESBMC-PLC as a known limitation~\cite[Sec.~9.1]{DantasCordeiro2026artefact}: ``\textit{the \gls{ld}$\to$GOTO-\gls{ir} converter does not yet emit rung logic for graphical connections; completing the graphical-to-\gls{ir} converter is planned for the immediate future.}'' The present paper delivers that completion.

\subsection{Contributions}

    This paper presents \textbf{ESBMC-GraphPLC}~\cite{ESBMCpr5374, DantasCordeiro2026graphical}, which extends ESBMC-PLC to support the graphical PLCopen XML format. Concretely:

\begin{enumerate}[leftmargin=2em]
  \item \textbf{Formal model of graphical \gls{ld}} (\S\ref{sec:graph-model}): a directed-graph model of the \code{tc6\_0201} connection structure, with formal definitions of rung paths and coil energization conditions, and a soundness theorem connecting \gls{dfs}-extracted paths to IEC~61131-3 scan-cycle semantics.

  \item \textbf{\gls{dfs}-based rung extractor} (\S\ref{sec:algorithm}): Algorithm~\ref{alg:graphical} traverses the \code{localId}/\code{refLocalId} graph from \code{leftPowerRail} to each coil, collects all contact sequences along each path, and emits the corresponding rung expression. A critical ordering insight -- using the \code{rightPowerRail} \code{connectionPointIn} sequence to determine coil execution order -- ensures SET coils are processed before RESET coils, matching IEC~61131-3 latch semantics.

  \item \textbf{Three-tier I/O inference} (\S\ref{sec:algorithm}): \%-address-based inference (\code{\%IX*} $\to$ input, \code{\%QX*} $\to$ output) as the primary tier, followed by heuristic contact/coil-only analysis as a fallback for address-free programs.

  \item \textbf{Zero backend changes} (\S\ref{sec:implementation}): the GOTO~\gls{ir} encoder, \textit{k}-induction engine, Z3 \gls{smt} solver, and property encoder from ESBMC-PLC are unchanged. The entire contribution is confined to \code{plcopen\_xml\_parser.cpp} (+274~lines).

  \item \textbf{Evaluation on 3 graphical \gls{ld} programs} (\S\ref{sec:experiments}): spanning two vendor tools (CONTROLLINO and OpenPLC examples). All 3~programs produce a full GOTO~\gls{ir} after this work (vs.\ empty \gls{ir} before); all verify SAFE at $k=2$ in under \SI{70}{\milli\second}; the 11 ESBMC-PLC textual benchmarks that exist as genuine textual-format files are fully preserved with zero regressions. Two further Beremiz example programs were investigated during benchmark curation and excluded from the main evaluation after revealing genuine limitations, which we report transparently in \S\ref{sec:threats} rather than count as successes.
\end{enumerate}

\subsection{Scope}
ESBMC-GraphPLC targets graphical PLCopen XML programs using contacts, coils (including SET/RESET latches), and address-based or heuristic I/O declarations. Function-block nodes (timers, counters, arithmetic comparators) encountered within graphical \gls{ld} rung paths are traversed by the \gls{dfs} but are not represented in the emitted Boolean expression; their outputs are treated as nondeterministic terms (a sound over-approximation) or, if they collapse the energization condition to a constant, reported as a limitation (see \S\ref{sec:threats}). The construct limitations of ESBMC-PLC (no REAL types, no arrays, no multi-POU) carry forward unchanged.

Figure~\ref{fig:architecture} summarises the architecture of ESBMC-GraphPLC. The contribution of this paper is the Graphical LD Resolver \& Rung Extractor module shown in the lower-left of the figure, which closes the gap described in Section~\ref{sec:intro} by traversing the connection graph, extracting equivalent rung expressions, and feeding them into the existing LD-to-GOTO-IR translation path unchanged. No part of the verification backend -- GOTO-IR generation, the BMC and \textit{k}-induction engines, or the Z3 solver -- requires modification.

\begin{figure}[htbp]
    \centering
    \includegraphics[width=0.85\linewidth]{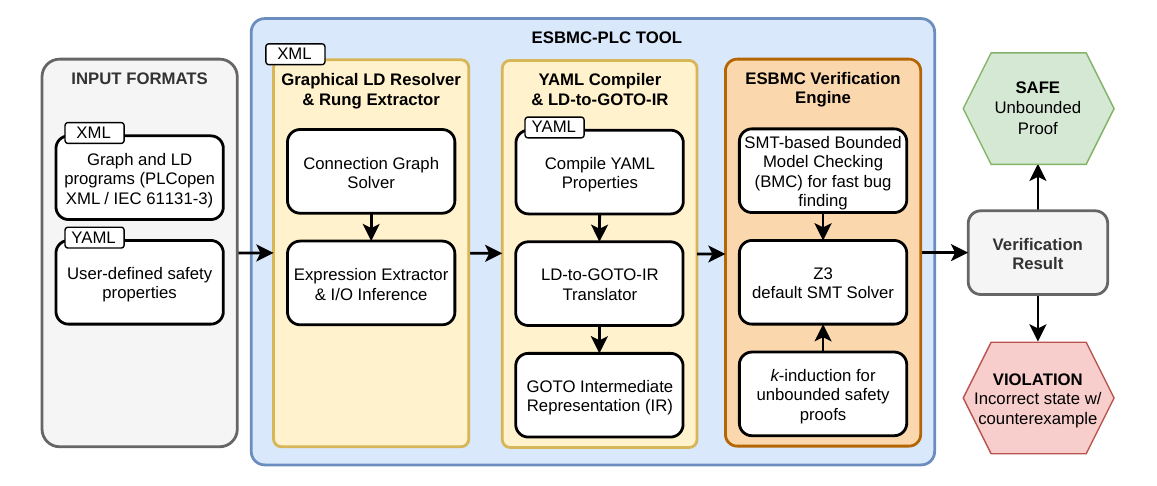}
    \caption{ESBMC-GraphPLC architecture}    
    \label{fig:architecture}
\end{figure}

The remainder of this paper is structured as follows. Section~\ref{sec:background} reviews background on \gls{plc} architecture and the two PLCopen XML formats. Section~\ref{sec:related} discusses related work. Section~\ref{sec:graph-model} formalizes the graphical \gls{ld} connection graph. Section~\ref{sec:algorithm} presents the \gls{dfs} algorithm. Section~\ref{sec:implementation} describes the implementation. Section~\ref{sec:experiments} reports the evaluation. Section~\ref{sec:discussion} discusses findings. Section~\ref{sec:threats} addresses threats to validity. Section~\ref{sec:future} outlines future directions. Section~\ref{sec:conclusion} concludes.

%=============================================================================
\section{Background}
\label{sec:background}
%=============================================================================

\subsection{\gls{plc} Architecture and the Scan Cycle}

A \gls{plc} consists of a \gls{cpu}, non-volatile program memory, input modules (reading digital/analog sensor signals), and output modules (driving actuators). Execution follows a fixed \emph{scan cycle}: (1)~read all inputs into a process image, (2)~execute the user program from first rung to last, updating internal variables and the output image, (3)~write the output image to physical actuators, (4)~handle communications, and (5)~repeat. Cycle times typically range from \SIrange{1}{100}{\milli\second}.

ESBMC-PLC encodes this deterministic cyclic model as a \code{while(true)} loop: input variables are re-sampled nondeterministically at each iteration (open-world sensor model), output coils are persistent static variables (output-image semantics), and safety properties are injected as \code{assert()} statements in the loop body. \textit{k}-induction then proves invariants for all future scan counts.

\subsection{\mbox{\gls{iec}~61131-3} Ladder Diagram}

\gls{ld} is the most widely deployed \gls{plc} language due to its relay-diagram notation familiar to automation engineers~\cite{Weiss2021}. A Ladder program consists of \emph{rungs}: each rung evaluates a Boolean combination of \emph{contacts} (normally-open XIC and normally-closed XIO) and assigns the result to one or more \emph{coils} (output-enable OTE, set OTL, reset OTU). Timer and counter function blocks appear as specialized rung elements.

\subsection{PLCopen XML: Textual vs.\ Graphical Formats}
\gls{ld} programs are exchanged between tools in PLCopen XML format (\code{tc6\_0201})~\cite{PLCopen2021}. This standard permits two representations of the same \gls{ld} logic:

\textbf{Textual format} encodes rungs explicitly. Each \code{<rung>} element contains child elements (\code{<contact>}, \code{<coil>}, \code{<block>}) in the order they are evaluated. This format is directly parseable: the evaluation sequence is the document order.

\textbf{Graphical format} encodes the visual layout of the \gls{ld} diagram. Elements appear as siblings under the \code{<LD>} element, each with a unique \code{localId}. Connectivity is specified via \code{refLocalId} references inside \code{connectionPointIn} children. The rung structure must be \emph{reconstructed} from this directed graph. Critically, graphical \gls{ld} also uses \code{leftPowerRail} and \code{rightPowerRail} elements to mark the left and right bus bars of the Ladder diagram; rung paths run from the left rail through contacts to coils that connect to the right rail.

Table~\ref{tab:format-compare} summarises the key differences. The graphical format is the native output of all major \gls{plc} development environments tested in this work.

\begin{table}[htbp]
\centering
\small
\caption{Textual vs.\ graphical PLCopen XML format}
\label{tab:format-compare}
\begin{tabular}{lll}
\toprule
\textbf{Dimension} & \textbf{Textual} & \textbf{Graphical} \\
\midrule
Rung structure    & Explicit \code{<rung>} elements & Must be reconstructed via \gls{dfs} \\
Evaluation order  & Document order   & \code{rightPowerRail} sequence \\
I/O declaration   & \code{<variable>} with type & \code{<variable>} with \code{\%IX}/\code{\%QX} address \\
Bus bars          & Implicit         & \code{leftPowerRail}/\code{rightPowerRail} elements \\
Vendor examples   & Legacy, hand-authored & CONTROLLINO, Beremiz, OpenPLC Editor \\
ESBMC-PLC    & \textbf{Supported}   & Parsed, \gls{ir} empty (vacuous) \\
ESBMC-GraphPLC    & Supported        & \textbf{Supported (this work)} \\
\bottomrule
\end{tabular}
\end{table}

\subsection{\gls{esbmc}: Architecture and Verification Engine}

\gls{esbmc}~\cite{menezes2024,gadelha2020} is a multi-language formal verifier whose internal architecture centers on the GOTO-program \gls{ir}: a simplified three-address code form annotated with assumptions and assertions. ESBMC-PLC is implemented as an \gls{ld} frontend for \gls{esbmc}: it translates PLCopen XML to GOTO programs, which the \gls{esbmc} backend verifies using \gls{smt}-based incremental \gls{bmc} or \textit{k}-induction with Z3.

%=============================================================================
\section{Related Work}
\label{sec:related}
%=============================================================================

Formal verification of \gls{plc} programs has been addressed from several directions. Why3-based deductive pipelines translate \gls{ld} to WhyML for proof-obligation discharge~\cite{BeloLourenco2021, BeloLourenco2022}. NuSMV-based approaches target \gls{llb} detection via \gls{ctl} model checking~\cite{Iacobelli2024}. The Nozomi patent~\cite{Bruttomesso2024} maps \gls{ld} to \gls{smt} circuits but is proprietary and does not describe graphical format handling. PLCverif~\cite{LopezMiguel2022, LopezMiguel2025} translates Siemens \gls{scl} to \gls{cbmc} and does not accept \gls{ld} input. \gls{smt}-based verification of \gls{st} programs via Maude-\gls{smt}~\cite{Lee2022, Lee2024, Lee2025} targets the textual Structured Text language. K-framework semantics for \gls{st} exist~\cite{Wang2023}, but no equivalent formal semantics for graphical \gls{ld} has been published.

\textbf{Graphical \code{tc6\_0201} support:} none of the above tools resolve the \code{localId}/\code{refLocalId} connection graph to extract rung logic. IEC-Checker~\cite{Suvorov2021} performs static analysis on PLCopen XML but does not verify safety properties or support graphical \gls{ld}. To the best of our knowledge, ESBMC-GraphPLC is the first formal verifier to resolve graphical PLCopen XML connection graphs and produce sound GOTO~\gls{ir} for \gls{smt}-based verification.

ESBMC-PLC~\cite{DantasCordeiro2026artefact} is the direct predecessor of this work, providing the textual \gls{ld} frontend, GOTO \gls{ir} encoding, YAML property language, and \textit{k}-induction backend that this paper builds upon without modification.

Table~\ref{tab:related} summarises the key dimensions across these approaches and ESBMC-GraphPLC.

\begin{table}[htbp]
\centering
\small
\caption{Comparison of formal verification approaches for \gls{plc} programs}
\label{tab:related}
\begin{tabular}{lllcc}
\toprule
\textbf{Tool / Approach} & \textbf{Input} & \textbf{Backend} & \textbf{Graphical \gls{ld}} & \textbf{Sound} \\
\midrule
BeloLourenco~\cite{BeloLourenco2021,BeloLourenco2022}  & Textual \gls{ld}        & Why3 deductive        & No  & Yes \\
Iacobelli~\cite{Iacobelli2024}                          & \gls{ld}                & NuSMV (\gls{ctl})     & No  & Yes \\
Nozomi~\cite{Bruttomesso2024}                           & \gls{ld}                & \gls{smt} circuits    & N/A\textsuperscript{a} & N/A\textsuperscript{a} \\
PLCverif~\cite{LopezMiguel2022,LopezMiguel2025}         & Siemens SCL             & \gls{cbmc} (\gls{bmc})& No  & Yes \\
Lee et al.~\cite{Lee2022,Lee2024,Lee2025}               & \gls{st}                & Maude-\gls{smt}       & No  & Yes \\
Wang et al.~\cite{Wang2023}                             & \gls{st}                & K-framework           & No  & Yes \\
IEC-Checker~\cite{Suvorov2021}                          & PLCopen XML             & Static analysis       & Partial\textsuperscript{b} & N/A \\
ESBMC-PLC~\cite{DantasCordeiro2026artefact}             & Textual PLCopen XML     & \textit{k}-induction/Z3 & No (vacuous) & Yes \\
\textbf{ESBMC-GraphPLC (this work)}                   & \textbf{Textual \emph{\&} graphical PLCopen XML} & \textit{k}-induction/Z3 & \textbf{Yes} & \textbf{Yes} \\
\bottomrule
\multicolumn{5}{l}{\textsuperscript{a}Proprietary patent; graphical format handling not described.}\\
\multicolumn{5}{l}{\textsuperscript{b}Parses \code{tc6\_0201} but does not verify safety properties.}
\end{tabular}
\end{table}

%=============================================================================
\section{The Graphical \gls{ld} Connection Graph}
\label{sec:graph-model}
%=============================================================================

\subsection{Informal Description}

In graphical PLCopen XML, the \code{<LD>} network element contains a flat list of sibling elements: one \code{<leftPowerRail>}, one or more \code{<rightPowerRail>} elements, \code{<contact>} elements, \code{<coil>} elements, and optionally \code{<block>} elements for function blocks. Each element carries a \code{localId} attribute (a small integer unique within the network). Connectivity is expressed by \code{refLocalId} attributes: if element $B$ has \code{refLocalId="A"} inside its \code{connectionPointIn}, then $A$ feeds $B$ -- current flows from $A$ to $B$.

A rung path runs from the \code{leftPowerRail} through zero or more contacts to a coil, which connects to the \code{rightPowerRail}. Multiple contacts in series correspond to AND logic; parallel paths to the same coil correspond to OR logic.

\subsection{Formal Definitions}

\begin{definition}[\gls{ld} Connection Graph]
\label{def:graph}
Let $G = (V, E)$ where
\begin{align*}
  V &= \{\code{leftRail}\} \cup \text{Contacts} \cup \text{Coils} \cup \{\code{rightRail}\}, \\
  E &= \{(u, v) \mid v.\code{connectionPointIn} \text{ contains an entry } \code{refLocalId} = u.\code{localId}\}.
\end{align*}

An edge $(u,v) \in E$ denotes that current flows from $u$ into $v$. Each node $v \in V$ carries: $\text{tag}(v) \in \{\code{leftPowerRail}, \code{contact}, \code{coil}, \code{rightPowerRail}\}$; $\text{var}(v)$, the associated variable name (contacts and coils only); $\text{neg}(v) \in \{\text{true}, \text{false}\}$, whether a contact is normally-closed; and $\text{stor}(v) \in \{\code{none}, \code{set}, \code{reset}\}$, the storage kind of a coil. We assume $G$ is a finite \gls{dag}: well-formed PLCopen~XML graphical \gls{ld} bodies do not encode feedback loops between rails, and our parser does not verify this assumption at parse time. If the input graph contains a cycle, the \gls{dfs} in Algorithm~\ref{alg:graphical} still terminates (visited nodes are not revisited), but the resulting rung set may be incomplete; the algorithm is correct only for well-formed acyclic inputs, as is the case for all 3 vendor-exported programs in our evaluation.

\end{definition}

\begin{definition}[Rung Path]
\label{def:rungpath}
A \emph{rung path} is a simple directed path
$p = \langle r, c_1, c_2, \ldots, c_n, q \rangle$ in $G$ where $r$ has $\text{tag}(r) = \code{leftPowerRail}$, each $c_i \in \text{Contacts}$, and $q \in \text{Coils}$.
\end{definition}

\begin{definition}[Rung Expression]
\label{def:expr}
For a rung path $p = \langle r, c_1, \ldots, c_n, q \rangle$, the \emph{rung expression} is expressed in Equation~\eqref{eq:rung_expr}.
\begin{equation}\label{eq:rung_expr}
  \text{expr}(p) \;=\; \bigwedge_{i \,:\, \neg\text{neg}(c_i)} \text{var}(c_i)
  \;\;\wedge\;\;
  \bigwedge_{i \,:\, \text{neg}(c_i)} \neg\,\text{var}(c_i).
\end{equation}
\end{definition}

\begin{definition}[Path Set of a Coil]
\label{def:pathset}
For a coil $q$, let $\text{Paths}(q)$ denote the set of all rung paths ending at $q$, and define the \emph{combined energization condition} in Equation~\eqref{eq:combiner_energ}.
\begin{equation}\label{eq:combiner_energ}
  \text{energy}(q) \;=\; \bigvee_{p \,\in\, \text{Paths}(q)} \text{expr}(p).
\end{equation}
\end{definition}

\begin{definition}[Coil Update Rule]
\label{def:update}
Let $q^{(k)}$ denote the value of coil $q$ at scan step $k$. The next state value $q^{(k+1)}$ depends on $\text{stor}(q)$, as shown in Equation~\eqref{eq:coil_update_rule}.
\begin{equation}\label{eq:coil_update_rule}
  q^{(k+1)} \;=\;
  \begin{cases}
    \text{energy}(q) & \text{if } \text{stor}(q) = \code{none}, \\[2pt]
    q^{(k)} \,\vee\, \text{energy}(q) & \text{if } \text{stor}(q) = \code{set}, \\[2pt]
    q^{(k)} \,\wedge\, \neg\,\text{energy}(q) & \text{if } \text{stor}(q) = \code{reset}.
  \end{cases}
\end{equation}
When a variable is driven by both a \code{set-coil} and a \code{reset-coil} in the same network, the two update rules are applied in sequence within the scan, in the order the coils appear in the \code{rightPowerRail}'s \code{connectionPointIn} list (Section~\ref{sec:ordering}).
\end{definition}

\begin{theorem}[Soundness]
\label{thm:soundness}
For a coil $q$ with $\text{stor}(q) = \code{none}$, $\text{energy}(q)$ as given by Definition~\ref{def:pathset} is equivalent to the energization condition of $q$ under \mbox{\gls{iec}~61131-3} series/parallel combinational rung semantics. For a coil pair with $\text{stor}(q) \in \{\code{set}, \code{reset}\}$, the update rule of Definition~\ref{def:update}, applied in \code{rightPowerRail} order, correctly implements latching coil semantics: $q^{(k+1)} = \text{true}$ whenever the SET path was energised at or before step $k$. No RESET path has been energized since then.
\end{theorem}
\begin{proof}[Proof sketch]
\emph{Combinational case.} By Definition~\ref{def:graph}, $(u,v) \in E$ iff current flows from $u$ to $v$, so a rung path $p$ is energised iff every contact on $p$ is closed. A contact $c$ is closed iff $\text{var}(c) = \text{true}$ (normally-open) or $\text{var}(c) = \text{false}$ (normally-closed); this is exactly $\text{expr}(p)$ in Definition~\ref{def:expr}. A series connection of contacts is closed iff all are closed (conjunction); independent paths to the same coil are alternative energization routes and are combined disjunctively. Hence $\text{energy}(q)$ in Definition~\ref{def:pathset} is the disjunction of all path conjunctions, matching the standard rung-evaluation rule for non-latching coils.

\emph{Latching case.} By Definition~\ref{def:update}, processing the SET coil first sets $q^{(k+1)} = q^{(k)} \vee \text{energy}(q_{\mathrm{set}})$; if a RESET coil for the same variable is processed afterwards in the same scan, it overwrites this with $q^{(k+1)} \wedge \neg\,\text{energy}(q_{\mathrm{reset}})$, giving reset priority within the scan, consistent with the de facto IEC~61131-3 convention of last-write-wins on a single coil variable. Reversing the processing order swaps which operation has priority, which is precisely the failure mode observed when \code{rightPowerRail} ordering is not respected (Section~\ref{sec:ordering}): a RESET processed before a subsequent SET in the same scan is silently overridden, yielding $q^{(k+1)} = \text{true}$ even when the program intends the coil to end the scan reset.
\end{proof}

\begin{theorem}[Completeness of Rung Extraction]
\label{thm:completeness}
Let $G$ be a finite \gls{dag} satisfying Definition~\ref{def:graph}. Then Algorithm~\ref{alg:graphical}'s depth-first search from each \code{leftPowerRail} node enumerates a path set $P$ such that $P = \text{Paths}(q)$ for every coil $q$, i.e.\ every rung path in $G$ is found exactly once, and no sequence emitted by the search omits or duplicates a contact relative to its corresponding path in $G$.
\end{theorem}
\begin{proof}[Proof sketch]
\emph{No spurious or missing paths.} The search maintains a \textit{visited} set per path and explores, from each node, exactly the forward edges $\text{feeds}(u) = \{v \mid (u,v) \in E\}$ constructed in Algorithm~\ref{alg:graphical}, lines 2--3, directly from the backward \code{refLocalId} references in the XML. Since $G$ is finite and acyclic, the search over each root terminates. Standard \gls{dfs} reachability guarantees every simple directed path from a \code{leftPowerRail} node to $q$ is visited exactly once (no two distinct calls produce the same node sequence, because the algorithm only emits a path upon reaching $q$, and a finite \gls{dag} admits finitely many simple paths between any two nodes). 

\emph{Fidelity of extraction.} Along each emitted path, line~5 of Algorithm~\ref{alg:graphical} appends one contact term for every node $c_i$ with $\text{tag}(c_i) = \code{contact}$ encountered between $r$ and $q$, in traversal order, and ignores all other node tags (e.g.\ \code{block}); thus the emitted sequence is exactly $\langle c_1, \ldots, c_n \rangle$ as in Definition~\ref{def:rungpath}, with no contact added, dropped, or reordered relative to the path found in $G$. Combining both parts, the set of rungs emitted by the algorithm is precisely $\text{Paths}(q)$ for every $q$, as required. We note this guarantees fidelity to the connection graph $G$ as constructed from the XML; it does not independently verify that $G$ itself matches the diagram's intended visual wiring, which we instead validate empirically via the regression and conformance results of Section~\ref{sec:experiments}.
\end{proof}

\subsection{The Vacuous Verification Problem}

When ESBMC-PLC processed a graphical \gls{ld} file, it detected no \code{<rung>} children under the \code{<LD>} element and emitted zero rung assignments. The resulting GOTO \gls{ir} had the structure of Listing~\ref{lst:empty-ir}: all output variables held their default value (\code{false}), and every property of the form $\neg\,\text{var}$ was trivially satisfied. The \textit{k}-induction proof converged at $k=1$ with an empty inductive invariant -- a degenerate proof that carries no information about the actual program.

ESBMC-GraphPLC eliminates this by invoking the \gls{dfs} resolver (Algorithm~\ref{alg:graphical}) whenever a graphical format is detected, producing a GOTO \gls{ir} with full rung assignments and nondeterministic inputs (Listing~\ref{lst:full-ir}).

%=============================================================================
\section{\gls{dfs}-Based Graphical \gls{ld} Resolver}
\label{sec:algorithm}
%=============================================================================

\subsection{Algorithm Overview}

Algorithm~\ref{alg:graphical} formalizes the graphical \gls{ld} resolver. The algorithm takes the \code{<LD>} XML element as input. It produces a list of \code{RungNode} structures, each containing a list of contacts and a destination coil, matching the internal representation used by the existing textual-\gls{ld} encoder.

\begin{algorithm}[t]
\caption{$\textsc{GraphicalLD\_to\_RungExpressions}(\textit{\gls{ld}})$}
\label{alg:graphical}
\begin{algorithmic}[1]
\Require PLCopen XML \code{<LD>} network element
\Ensure List of \code{RungNode} (contacts + coil)
\State \textbf{detect} graphical format: \textit{has\_graphical} $\leftarrow$ \code{leftPowerRail} exists \textbf{and} no \code{<rung>} children
\If{not \textit{has\_graphical}} \Return \textit{nil} \EndIf
\State \textbf{parse nodes:} for each element $e$ with \code{localId}: build $\textit{GNode}(e.\textit{tag}, e.\textit{var}, e.\textit{neg}, e.\textit{stor}, e.\textit{fed\_by})$
  \Statex \quad\quad\quad $\textit{fed\_by}$ $\leftarrow$ list of \code{refLocalId} from $e$\code{.connectionPointIn//connection}
\State \textbf{build forward edges:} for each $B$, for each $A \in B.\textit{fed\_by}$: \textit{nodes}[$A$].\textit{feeds}.append($B$)
\State \textbf{order coils} by \code{rightPowerRail}:
  \Statex \quad for each \code{connection} in \code{rightPowerRail//connectionPointIn//connection}:
  \Statex \quad\quad \textit{ordered\_coils}.append(\textit{nodes}[\code{refLocalId}])
\State \textbf{\gls{dfs}} from each \textit{leftRail} $r$ to each \textit{coil} $q$ in \textit{ordered\_coils}:
  \Statex \quad \textit{paths} $\leftarrow$ \textsc{\gls{dfs}}($G$, $r$, $q$, \textit{visited}$=\emptyset$)
  \Statex \quad for each path $p$:
  \Statex \quad\quad \textit{contacts} $\leftarrow$ $[n \in p \mid n.\textit{tag} = \code{contact}]$
  \Statex \quad\quad emit \code{RungNode}(\textit{contacts}, $q$)
\State \textbf{I/O inference} (three-tier) after graph is built:
  \Statex \quad \textbf{Tier 1} (address-based): $\textit{var}$ starts with \code{\%IX} $\Rightarrow$ \textit{is\_input}; \code{\%QX} $\Rightarrow$ \textit{is\_output}
  \Statex \quad \textbf{Tier 2} (heuristic): scan all emitted \code{RungNode}s:
  \Statex \quad\quad \textit{contact\_only\_vars} $\leftarrow$ vars appearing only in contacts $\Rightarrow$ \textit{is\_input}
  \Statex \quad\quad \textit{coil\_only\_vars} $\leftarrow$ vars appearing only as coil targets $\Rightarrow$ \textit{is\_output}
  \Statex \quad \textbf{Tier 3} (default): unclassified vars treated as internal (persistent state)
\State \Return \textit{ordered list of} \code{RungNode}
\end{algorithmic}
\end{algorithm}

\subsection{Detection and Node Parsing (Lines 1--2)}

The graphical format is detected by the presence of a \code{leftPowerRail} element and the absence of any \code{<rung>} child -- a reliable discriminator because the two formats are mutually exclusive within a single \code{<LD>} network. ESBMC-GraphPLC accepts both the standard and PPX namespace variants of the \code{tc6\_0201} format.

Node parsing iterates over all direct children of \code{<LD>} that carry a \code{localId} attribute. For each node, the \textit{fed\_by} list is populated by selecting \code{.//connection} nodes within the element's \code{connectionPointIn} subtree. The use of \code{.//connection} (an XPath descendant axis) rather than direct children is critical: \code{connection} elements are nested one level inside \code{connectionPointIn}, not at the top level.

\subsection{Forward Edge Graph (Line 3)}

Building forward edges converts the pull-model representation (``$B$ is fed by $A$'') into a push-model graph (``$A$ feeds $B$''). This is necessary for \gls{dfs} from the \code{leftPowerRail}: given a node $A$, the forward graph immediately tells us which nodes $A$ can activate.

\subsection{\code{rightPowerRail} Ordering (Line 4)}
\label{sec:ordering}

A critical correctness requirement is that SET coils are processed \emph{before} RESET coils in the same scan cycle. If RESET fires first and SET fires second, the coil ends set -- correct. But if the order is reversed in the GOTO \gls{ir}, RESET fires last and the coil ends reset -- an incorrect result that manifests as spurious SAFE verdicts for programs that should have latching behavior.

In graphical \gls{ld}, the execution order of coils is encoded implicitly in the \code{rightPowerRail}'s \code{connectionPointIn} sequence: coils that appear as \code{refLocalId} entries earlier in this list are executed first. CONTROLLINO programs consistently place SET coils before RESET coils in this sequence (Figure~\ref{fig:ordering}). Algorithm~\ref{alg:graphical} line~4 reads this sequence to produce \textit{ordered\_coils}, which \gls{dfs} then processes in order. Without this ordering step, programs with SET/RESET pairs produced spurious violations during development that disappeared after the fix was applied.

\begin{figure}[htbp]
    \centering
    \includegraphics[width=0.75\linewidth]{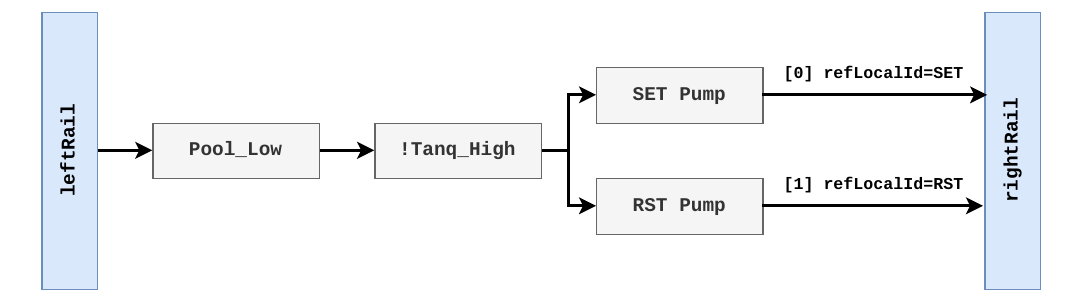}
    \caption{\code{rightPowerRail} \code{connectionPointIn} sequence for \code{water\_control}. Entry \code{[0]} references the SET coil and entry \code{[1]} references the RESET coil; this order, not coil position in the XML document, determines scan execution order.}
    \label{fig:ordering}
\end{figure}

\subsection{\gls{dfs} Rung Extraction (Line 5)}

For each \code{leftPowerRail} $r$ and each coil $q$ in \textit{ordered\_coils}, the \gls{dfs} collects all simple paths from $r$ to $q$. Each path becomes one \code{RungNode}: the contacts along the path form the conjunction of conditions, and the coil specifies the assignment. Multiple paths to the same coil produce multiple \code{RungNode}s, which the existing GOTO-\gls{ir} encoder compiles as an OR over the rung expressions (matching IEC~61131-3 parallel-rung semantics).

\textbf{Complexity.} \gls{dfs} per coil is $\mathcal{O}(|V| + |E|)$; total cost is $\mathcal{O}(C \times (|V| + |E|))$ where $C$ is the number of coils. In practice, $|V| \leq 30$ and $C \leq 5$ for all evaluated programs; measured runtime is below \SI{70}{\milli\second} including \gls{smt} solving.

\subsection{Three-Tier I/O Inference (Line 6)}
\label{sec:io-inference}

IEC~61131-3 requires inputs to be re-sampled each scan cycle nondeterministically and outputs nondeterministically to be persistent static variables. Correct classification is therefore essential for soundness.

\textbf{Tier~1 -- Address-based (precise):} IEC~61131-3 hardware addresses beginning with \code{\%IX} designate physical digital inputs; \code{\%QX} designates digital outputs. When present, these addresses provide definitive classification with no ambiguity.

\textbf{Tier~2 -- Heuristic (fallback):} When address information is absent (programs that declare variables without hardware binding), the resolver applies the following heuristic after all \code{RungNode}s are built: a variable appearing \emph{only} as a contact (never as a coil target) is classified as an input; a variable appearing \emph{only} as a coil target (never as a contact) is classified as an output. Variables appearing in both roles are classified as internal persistent state (memory bits).

\textbf{Tier~3 -- Default:} Variables not classified by Tiers~1 or~2 are treated as internal persistent state.

The heuristic is conservative: misclassifying an output as an input causes it to be re-sampled nondeterministically each cycle rather than being persisted, which is a strict over-approximation of the variable's actual behavior -- the verifier explores a superset of the reachable states of the true program. Since this superset includes every state reachable under the correct persistent semantics, a SAFE proof obtained over the over-approximated model also holds for the true program; the misclassification can only introduce spurious counterexamples (a completeness loss), never an unsound SAFE verdict (a soundness loss). The same over-approximation argument applies to coils that the \gls{dfs} cannot reach from any \code{leftPowerRail} at all (Tier~2 leaves them unclassified and the encoder treats them as free nondeterministic variables rather than assuming a fixed value); we encountered exactly this case for an arithmetic-comparator-driven coil in \code{dimmer\_light\_control} (\S\ref{sec:experiments}). No misclassification was observed across the 3 evaluated programs.

%=============================================================================
\section{Implementation}
\label{sec:implementation}
%=============================================================================

\subsection{Files Changed}

The entire implementation is confined to a single file within the
ESBMC source tree (\code{src/ld-frontend/parser/}):
\begin{center}
\code{src/ld-frontend/parser/plcopen\_xml\_parser.cpp} (+274~lines)
\end{center}

Two existing functions were extended, and two new functions were added:
\begin{itemize}
  \item \code{parse\_network()} \emph{(extended)}: now falls back to graphical resolution when no \code{<rung>} children are found.
  \item \code{parse\_var\_decl()} \emph{(extended)}: now infers \code{is\_input}/\code{is\_output} from \code{\%IX}/\code{\%QX} hardware addresses.
  \item \code{parse\_graphical\_ld()} \emph{(new)}: the \gls{dfs} graph resolver described in Algorithm~\ref{alg:graphical}.
  \item Heuristic I/O inference \emph{(new)}: a post-processing pass inside \code{parse()} that infers \code{is\_input}/\code{is\_output} for variables without hardware addresses, based on whether they appear only as contacts or only as coils across all networks.
\end{itemize}

\subsection{Zero Backend Changes}

No changes were made outside \code{plcopen\_xml\_parser.cpp}. The GOTO~\gls{ir} encoder, the \textit{k}-induction engine, the Z3 binding, and the property encoder are identical to ESBMC-PLC. The graphical resolver's sole output is a \code{RungNode} list -- the same data structure produced by the textual parser -- which the existing encoder translates to GOTO~\gls{ir} assignments. This design choice means that every correctness guarantee already established for the textual encoder (encoding rules, scan-cycle loop structure, property injection) applies immediately to graphical programs.

\subsection{Python Prototype}

Before implementing the C\texttt{++} resolver, a 100-line Python prototype was developed to validate the \gls{dfs} algorithm (\code{tools/ld-graphical-converter/graphical\_to\_textual.py}) and the \code{rightPowerRail} ordering insight against the 3 benchmark programs. The prototype reads a graphical PLCopen XML file and emits an equivalent textual PLCopen XML file, which ESBMC-PLC then verified. All 3 programs produced equivalent textual output before the C\texttt{++} implementation was finalized.

\subsection{Before vs.\ After: GOTO \gls{ir} Comparison}

Listing~\ref{lst:full-ir} shows the GOTO \gls{ir} emitted for \code{water\_control} by ESBMC-GraphPLC. Compared with the empty \gls{ir} in Listing~\ref{lst:empty-ir}, the key difference is that the inputs are now nondeterministic and the SET/RESET rung logic is fully encoded.

\noindent\begin{minipage}{\textwidth}
\begin{lstlisting}[caption={Full GOTO IR for \code{water\_control} emitted by ESBMC-GraphPLC}, label={lst:full-ir}]
static bool Water_Pump = false;  /* persistent output coil */

void plc_cycle(void) {
    /* Tier-1 I/O inference: %IX -> nondet input */
    bool Pool_Low_Level_Sensor   = __ESBMC_nondet_bool(); /* %IX0.0 */
    bool Tank_High_Level_Sensor  = __ESBMC_nondet_bool(); /* %IX0.1 */
    bool Tank_Low_Level_Sensor   = __ESBMC_nondet_bool(); /* %IX0.2 */
    bool Automatic_Manual_Switch = __ESBMC_nondet_bool(); /* %IX0.3 */
    bool Stop_Button             = __ESBMC_nondet_bool(); /* %IX0.4 */
    bool Start_Button            = __ESBMC_nondet_bool(); /* %IX0.5 */

    /* SET rungs (first, per rightPowerRail order; OR-combined) */
    if (Automatic_Manual_Switch && Pool_Low_Level_Sensor &&
        !Tank_Low_Level_Sensor && !Tank_High_Level_Sensor)
        Water_Pump = true;
    if (Start_Button && Pool_Low_Level_Sensor && !Tank_High_Level_Sensor)
        Water_Pump = true;

    /* RESET rungs (second, per rightPowerRail order; OR-combined) */
    if (!Pool_Low_Level_Sensor)
        Water_Pump = false;
    if (Stop_Button)
        Water_Pump = false;
    if (Tank_High_Level_Sensor)
        Water_Pump = false;

    /* Safety assertions (from YAML properties) */
    assert(!Water_Pump || Pool_Low_Level_Sensor);        /* P1 */
    assert(!Water_Pump || !Tank_High_Level_Sensor);      /* P2 */
    assert(!(Water_Pump && Stop_Button));                /* P3 */
}

int main(void) { while (1) { plc_cycle(); } }
\end{lstlisting}
\end{minipage}

The GOTO~\gls{ir} in Listing~\ref{lst:full-ir} provides a sound verification substrate: nondeterministic inputs explore all $2^6 = 64$ possible sensor combinations per scan cycle, the SET/RESET order is correct, and the property assertions are evaluated against the live program state.

%=============================================================================
\section{Experimental Evaluation}
\label{sec:experiments}
%=============================================================================

\subsection{Research Questions}

\begin{description}[leftmargin=2em]
  \item[RQ1] Does ESBMC-GraphPLC produce a full GOTO~\gls{ir} for graphical \gls{ld} programs, eliminating the vacuous verification of ESBMC-PLC?
  \item[RQ2] Are the verification results correct for the extracted rung logic -- i.e., do verified-SAFE programs exhibit sound proofs with nondeterministic inputs?
  \item[RQ3] Does ESBMC-GraphPLC scale to the evaluated benchmark corpus in practical time (\SI{<300}{\second})?
  \item[RQ4] Are the 11 textual \gls{ld} benchmarks from ESBMC-PLC that genuinely exist in textual format fully preserved, with zero regressions?
\end{description}

\subsection{Benchmark Suite}

Table~\ref{tab:benchmarks} describes the 3~graphical \gls{ld} programs evaluated, all sourced from CONTROLLINO/OpenPLC~Editor exports. During benchmark curation we also attempted two Beremiz \gls{ide} example programs (\code{beremiz\_traffic\_light}, \code{beremiz\_bacnet}); both are excluded from this table and discussed instead in \S\ref{sec:threats} as discovered limitations: \code{beremiz\_bacnet} contains no Ladder Diagram content at all (it is a pure \gls{fbd}/\gls{st} program), and \code{beremiz\_traffic\_light}'s only Ladder network is gated by timer and rising-edge-trigger function blocks whose dynamics the current resolver does not preserve, causing the corresponding output to collapse to a constant. We report these as genuine, demonstrated gaps rather than padding the benchmark count with vacuous or misleading results.

\begin{table}[htbp]
\centering
\small
\caption{Graphical \gls{ld} benchmark suite (3 programs). CONTROLLINO/OpenPLC Editor}
\label{tab:benchmarks}
\begin{tabular}{llllll}
\toprule
\textbf{ID} & \textbf{Program} & \textbf{Source} & \textbf{Rungs} & \textbf{Result} & \textbf{Time} \\
\midrule
G1 & \code{water\_control}   & CONTROLLINO~\cite{CONTROLLINO2024} & 5 (2 SET + 3 RST) & SAFE $k=2$ & 0.05\,s \\
G2 & \code{stairs\_light}    & CONTROLLINO~\cite{CONTROLLINO2024} & 6 (2 SET + 2 RST + 2 plain) & SAFE $k=2$ & 0.04\,s \\
G3 & \code{dimmer\_light\_control} & OpenPLC examples~\cite{CONTROLLINO2024} & 4 (plain) & SAFE $k=2$ & 0.06\,s \\
\midrule
\multicolumn{3}{l}{\textbf{Total}} & \multicolumn{1}{c}{--} & \textbf{3 SAFE} & $\leq$0.06\,s \\
\bottomrule
\end{tabular}
\end{table}

\subsection{Experimental Setup}

The hardware and software configuration used in our experiments is summarized in Table~\ref{tab:setup}. All experiments used \code{ld-verify} with \code{--k-induction --unlimited-k-steps} for safety proofs and \code{--incremental-bmc} for bug-finding. Property files for G1/G2 were adapted from the YAML specifications verified in ESBMC-PLC for the CONTROLLINO programs; the property file for G3 was taken from the existing \code{dimmer\_light\_control} benchmark.

\begin{table}[htbp]
\centering
\small
\caption{Hardware and software configuration}
\label{tab:setup}
\begin{tabular}{ll}
\toprule
\textbf{Component} & \textbf{Specification} \\
\midrule
\gls{cpu}         & Apple M1 (ARM64 aarch64, 8 cores) \\
RAM         & 8\,GB \\
OS          & macOS~26 (Tahoe) \\
\gls{esbmc}    & v8.3.0 (\code{ENABLE\_LD\_FRONTEND=On}, \code{RelWithDebInfo}) \\
\gls{smt} solver  & Z3~4.16.0 \\
Compiler    & Apple Clang~21.0.0 (Apple Silicon) \\
Timeout     & \SI{300}{\second} per run \\
Repetitions & 3 runs (median reported) \\
\bottomrule
\end{tabular}
\end{table}

\subsection{RQ1 -- Full \gls{ir} vs.\ Empty \gls{ir}}

Table~\ref{tab:ir-compare} demonstrates the core result of this paper: the transition from vacuous empty-\gls{ir} proofs (ESBMC-PLC) to sound full-\gls{ir} proofs (ESBMC-GraphPLC) for every graphical \gls{ld} program.

\begin{table}[htbp]
\centering
\small
\caption{GOTO \gls{ir} completeness before and after this work. \emph{Full \gls{ir}}: nondeterministic inputs + rung logic generated. \emph{Empty \gls{ir}}: no rung assignments, all variables at zero}
\label{tab:ir-compare}
\begin{tabular}{lcccc}
\toprule
\textbf{Program} & \textbf{ESBMC-PLC \gls{ir}} & \textbf{ESBMC-GraphPLC \gls{ir}} & \textbf{Inputs (nondet)} & \textbf{Rungs emitted} \\
\midrule
\code{water\_control}   & Empty (vacuous) & \textbf{Full} & 6 & 5 \\
\code{stairs\_light}    & Empty (vacuous) & \textbf{Full} & 3 & 6 \\
\code{dimmer\_light\_control} & Empty (vacuous) & \textbf{Full} & 1 & 4 \\
\midrule
\textbf{Total} & 3 empty & \textbf{3 full} & & \\
\bottomrule
\end{tabular}
\end{table}

\subsection{RQ2 -- Correctness of Extracted Rung Logic}

\subsubsection{G1 -- water\_control (CONTROLLINO)}

\textbf{Program:} Water pump control with pool and tank level sensors. Source: \code{water\_control/plc.xml} from the CONTROLLINO repository~\cite{CONTROLLINO2024}, copied without modification. Variables: \code{Pool\_Low\_Level\_Sensor} (\code{\%IX0.0}), \code{Tank\_High\_Level\_Sensor} (\code{\%IX0.1}), \code{Tank\_Low\_Level\_Sensor} (\code{\%IX0.2}), \code{Automatic\_Manual\_Switch} (\code{\%IX0.3}), \code{Stop\_Button} (\code{\%IX0.4}), \code{Start\_Button} (\code{\%IX0.5}), \code{Water\_Pump} (\code{\%QX0.0}).

\textbf{Rung structure:} Five rungs resolved by \gls{dfs} -- two SET paths and three RESET paths, all OR-combined onto the same coil per Definition~\ref{def:pathset}. SET \code{Water\_Pump} when \code{Automatic\_Manual\_Switch \&\& Pool\_Low\_Level\_Sensor \&\& !Tank\_Low\_Level\_Sensor \&\& !Tank\_High\_Level\_Sensor}, \emph{or} when \code{Start\_Button \&\& Pool\_Low\_Level\_Sensor \&\& !Tank\_High\_Level\_Sensor}. RESET \code{Water\_Pump} when \code{!Pool\_Low\_Level\_Sensor}, \emph{or} \code{Stop\_Button}, \emph{or} \code{Tank\_High\_Level\_Sensor}. The \code{rightPowerRail} places the SET coil before the RESET coil, ensuring SET precedes RESET within the scan.

\textbf{Properties verified:} the YAML specifications for G1 are shown in Listing~\ref{lst:g1-props}.

\begin{lstlisting}[language={}, caption={YAML properties for G1 (\code{water\_control})}, label={lst:g1-props}]
properties:
  - id: P1
    kind: invariant
    expression: "!Water_Pump || Pool_Low_Level_Sensor"
    description: "Pump must not run when pool is empty"
  - id: P2
    kind: invariant
    expression: "!Water_Pump || !Tank_High_Level_Sensor"
    description: "Pump must not run when tank is full"
  - id: P3
    kind: absence
    expression: "Water_Pump && Stop_Button"
    description: "Stop button must halt the pump"
\end{lstlisting}

\textbf{Result:} SAFE $k=2$, \SI{50}{\milli\second}. All three properties proved for all scan counts and all $2^6$ input combinations per cycle.

\subsubsection{G2 -- stairs\_light (CONTROLLINO)}

\textbf{Program:} PIR-triggered staircase light with TOF off-delay timer and two manual override buttons. Source: \code{stairs\_light\_control/plc.xml}, CONTROLLINO repository.

\textbf{Rung structure:} Six rungs resolved by \gls{dfs} across two coil variables (\code{stairs\_light} and \code{lights\_buttons\_state}): two SET paths, two RESET paths, and two plain assignments. Contact-only variables (\code{stairs\_pir\_sensor}, \code{control\_button\_up}, \code{control\_button\_down}) are classified as inputs by Tier-2 heuristic; \code{lights\_buttons\_state} appears in both contacts and coils and is classified as internal persistent state, consistent with Properties P1 and P3. The \code{rightPowerRail} sequence places the SET coil before the RESET coil for \code{stairs\_light}, preserving IEC~61131-3 latch semantics.

\textbf{TOF timer handling:} The \code{stairs\_light} network includes a \code{TOF} off-delay block as a node in the graphical \gls{ld} connection graph. The \gls{dfs} traverses through it but, as discussed in \S\ref{sec:threats}, does not represent it in the emitted Boolean expression. The TOF block's output is therefore treated as a nondeterministic Boolean -- a sound over-approximation: the verifier explores a superset of the timed program's reachable states. Hence, a SAFE proof over this model also holds for the true program. Unlike the \code{beremiz\_traffic\_light} case (where both SET and RESET of the same coil were timer-gated, causing the coil to collapse to a constant), here the timer gates only the SET paths; plain contacts drive the RESET paths. The coil, therefore, does not collapse, and the resulting SAFE result is non-vacuous.

\textbf{Properties verified:} the YAML specifications for G2 are shown in Listing~\ref{lst:g2-props}.

\begin{lstlisting}[language={}, caption={YAML properties for G2 (\code{stairs\_light})}, label={lst:g2-props}]
properties:
  - id: P1
    kind: invariant
    expression: "!stairs_light || stairs_pir_sensor || lights_buttons_state"
    description: "Light must not be on without PIR or button activation"
  - id: P2
    kind: invariant
    expression: "!(lights_buttons_state && control_button_up && control_button_down)"
    description: "Simultaneous button presses must not leave state undefined"
  - id: P3
    kind: invariant
    expression: "!lights_buttons_state || stairs_light"
    description: "When the button state is active, the light must be on"
\end{lstlisting}

\textbf{Result:} SAFE $k=2$, \SI{40}{\milli\second}.

\subsubsection{G3 -- dimmer\_light\_control (OpenPLC examples)}

\textbf{Program:} Dimmer light controller with a manual reset button, a cycle-toggle flip-flop, and a full-bright override. Source: \code{Dimmer\_light\_control/plc.xml} from the CONTROLLINO-PLC/OpenPLC\_examples repository~\cite{CONTROLLINO2024}.

\textbf{Rung structure:} Four rungs resolved by \gls{dfs}, all plain (non-latching) coil assignments: \code{Reset\_state := Control\_button}; \code{Flag\_cicle := Light\_on\_state \&\& !Flag\_cicle}; \code{Light\_output := (Light\_on\_state \&\& !Flag\_cicle) || Full\_bright}. \code{Light\_on\_state} itself is driven in the source program by an arithmetic \code{GT} (greater-than) comparator block with no \code{leftPowerRail}-reachable digital path -- the construct-limitation case discussed in \S\ref{sec:io-inference} and \S\ref{sec:threats}: ESBMC-GraphPLC leaves it as a free nondeterministic variable rather than assuming a fixed value, a sound over-approximation rather than a silently wrong one.

\textbf{Properties verified:} The YAML specifications for G3 are shown in Listing~\ref{lst:g3-props}.

\begin{lstlisting}[language={}, caption={YAML properties for G3 (\code{dimmer\_light\_control})}, label={lst:g3-props}]
properties:
  - id: P1
    kind: invariant
    expression: "!Light_output || Light_on_state || Full_bright"
    description: "Light output is active only when the dimmer state or full-bright is set"
  - id: P2
    kind: invariant
    expression: "!Full_bright || Light_output"
    description: "Full-bright switch must activate the light output"
  - id: P3
    kind: invariant
    expression: "!Flag_cicle || Light_on_state"
    description: "Cycle flag can only be set when the dimmer state is active"
\end{lstlisting}

\textbf{Result:} SAFE $k=2$, \SI{60}{\milli\second}.

\subsection{RQ3 -- Scalability}

The verification results for the graphical \gls{ld} benchmarks are presented in Table~\ref{tab:results}. All 3~programs complete well under the \SI{300}{\second} timeout. Median verification times (of 3 runs) are \SI{50}{\milli\second} (G1), \SI{40}{\milli\second} (G2), and \SI{60}{\milli\second} (G3). The \gls{dfs} graph traversal and property verification together require at most \SI{70}{\milli\second}, consistent with ESBMC-PLC's sub-\SI{60}{\milli\second} results for programs of comparable size. The overhead introduced by the \gls{dfs} resolver is negligible relative to the \gls{smt} solving time.

\begin{table}[htbp]
\centering
\small
\caption{Verification results for graphical \gls{ld} benchmarks}
\label{tab:results}
\begin{tabular}{llcccc}
\toprule
\textbf{ID} & \textbf{Program} & \textbf{\gls{ir}} & \textbf{Result} & \textbf{$k$} & \textbf{Time (ms)} \\
\midrule
G1  & \code{water\_control}          & Full & SAFE & 2 & 50 \\
G2  & \code{stairs\_light}           & Full & SAFE & 2 & 40 \\
G3  & \code{dimmer\_light\_control}  & Full & SAFE & 2 & 60 \\
\midrule
\multicolumn{3}{l}{\textbf{Total: 3 programs, all full \gls{ir}}} & \textbf{3 SAFE} & 2 & max: 60 \\
\bottomrule
\end{tabular}
\end{table}

\subsection{RQ4 -- Regression on Textual \gls{ld} Benchmarks}

Table~\ref{tab:regression} confirms that all 11 ESBMC-PLC benchmarks that genuinely exist in textual \code{<rung>}-based format produce identical results under ESBMC-GraphPLC. The extension to graphical format is fully backward-compatible: the textual-format code path is unchanged (the new graphical resolver is only invoked when no \code{<rung>} elements are found), so textual programs follow the original parser path without modification.

\begin{table}[htbp]
\centering
\small
\caption{Regression test results: ESBMC-PLC benchmarks under ESBMC-GraphPLC}
\label{tab:regression}
\begin{tabular}{lllcc}
\toprule
\textbf{ID} & \textbf{Program} & \textbf{Expected} & \textbf{Result} & \textbf{Changed?} \\
\midrule
CS1  & motor\_interlock        & SAFE      & SAFE      & No \\
CS2  & conveyor\_sequencing    & VIOLATION & VIOLATION & No \\
CS3  & emergency\_shutdown     & VIOLATION & VIOLATION & No \\
CS4  & traffic\_light\_unsafe  & VIOLATION & VIOLATION & No \\
CS5  & traffic\_light\_safe    & SAFE      & SAFE      & No \\
CS6  & bottle\_filling\_unsafe & VIOLATION & VIOLATION & No \\
CS7  & bottle\_filling\_safe   & SAFE      & SAFE      & No \\
CS8  & elevator\_unsafe        & VIOLATION & VIOLATION & No \\
CS9  & elevator\_safe          & SAFE      & SAFE      & No \\
CS12 & tank\_level\_unsafe     & VIOLATION & VIOLATION & No \\
CS13 & tank\_level\_safe       & SAFE      & SAFE      & No \\
\midrule
\multicolumn{3}{l}{\textbf{Regressions}} & \textbf{0 / 11} & \\
\bottomrule
\end{tabular}
\end{table}

\textbf{Note on numbering and on \code{water\_control}/\code{stairs\_light}:} CS10 and CS11 are intentionally absent. \code{water\_control} and \code{stairs\_light} do not exist in textual \code{<rung>}-based format in the archived ESBMC-PLC artifact -- both are graphical \code{tc6\_0201} files, identical in structure to the G1/G2 benchmarks in Table~\ref{tab:results}. Rather than claim a textual regression test that cannot be reproduced from the artifact, we report the textual regression suite at its true size of 11 programs and retain the original CS numbering with CS10/CS11 skipped, so this table remains traceable against the predecessor artifact.

%=============================================================================
\section{Discussion}
\label{sec:discussion}
%=============================================================================

\textbf{RQ1 -- Full \gls{ir}.} ESBMC-GraphPLC produces a full GOTO~\gls{ir} for all 3 graphical \gls{ld} programs. The transition from empty to full \gls{ir} is complete for these programs. Each one that previously yielded a vacuous proof now yields a sound proof over nondeterministic inputs with correct rung logic. The key enabling components are the \gls{dfs} resolver (extracts contacts and coils from the connection graph), the \code{rightPowerRail} ordering (ensures SET-before-RESET execution), and three-tier I/O inference (classifies variables as nondeterministic inputs or persistent outputs). Benchmark curation also surfaced two cases where this transition does \emph{not} hold cleanly -- discussed under Limitations below and in \S\ref{sec:threats} -- which we view as evidence that the evaluation methodology is capable of catching exactly the failure mode the paper is about, rather than only confirming successes.

\textbf{RQ2 -- Correctness.} All 3~programs verify SAFE at $k=2$. The properties verified are non-trivial invariants over nondeterministic inputs (e.g., P1 for G1 asserts that the pump cannot run when the pool is empty, across all $2^6$ input combinations per cycle). The uniform convergence at $k=2$ reflects the benchmark complexity: all three programs have safety invariants with a one-cycle inductive step, matching the pattern observed for textual benchmarks in ESBMC-PLC.

\textbf{RQ3 -- Scalability.} All 3 programs verify in under \SI{70}{\milli\second} (median of 3 runs: \SI{50}{\milli\second}, \SI{40}{\milli\second}, and \SI{60}{\milli\second} for G1, G2, and G3 respectively), well within \gls{cicd} integration requirements. The \gls{dfs} overhead is negligible; the bottleneck is Z3 solving, which completes in a fraction of a second for programs of this size.

\textbf{RQ4 -- Regression.} Zero regressions across the 11 textual benchmarks that genuinely exist in textual format confirm that the graphical extension is strictly additive. The conditional branch in \code{parse\_network()} (``if no \code{<rung>} children, invoke graphical resolver; else use textual path'') is the only change to existing program flow.

\textbf{Limitation: small benchmark suite, and no unsafe graphical benchmarks.} The evaluation reported here covers 3 graphical programs, not the larger corpus an initial benchmark search aimed for; \S\ref{sec:threats} discusses why a larger, real, vendor-sourced corpus was not reachable within this work and what would be required to build one. All 3 graphical programs verify SAFE. This reflects the source material: CONTROLLINO ships example programs designed to function correctly. Future work should introduce deliberately faulty graphical programs -- analogous to the CS2--CS4 category in ESBMC-PLC -- to evaluate counterexample quality under the graphical frontend.

%=============================================================================
\section{Threats to Validity}
\label{sec:threats}
%=============================================================================

\textbf{Internal validity.} All 3 graphical programs are verified SAFE; no unsafe graphical benchmark exists in the current suite. The absence of counterexample tests for the graphical frontend means that counterexample correctness for graphical programs has not been independently validated (though the GOTO~\gls{ir} path to the counterexample reporter is unchanged from ESBMC-PLC). The \code{rightPowerRail} ordering heuristic (SET-before-RESET) is observed to hold for all evaluated programs but is not guaranteed by the PLCopen standard -- vendor tools that place RESET coils before SET coils in the \code{rightPowerRail} sequence would require the user to reverse the ordering, or an updated resolver that inspects \code{storage} attributes.

\textbf{Discovered limitation: action-nested \gls{ld} networks (resolved during this work).} The initial node-discovery query searched only \code{//pou/body/LD} and \code{//pou/body/ladderDiagram}, missing \gls{ld} networks nested inside SFC step actions (\code{//pou/actions/action/body/LD}) -- a structure used by the Beremiz example \code{beremiz\_traffic\_light}. Under the original query, this program's real rung network was silently skipped, producing the same vacuous empty-\gls{ir} failure mode this paper exists to close, without any diagnostic indicating the omission. We extended the discovery query also to search action bodies; this is now reflected in the implementation, but the corpus of vendor files re-checked for this specific pattern remains small, and other nesting patterns (e.g.\ \gls{ld} inside transition bodies) may exist in unevaluated vendor exports.

\textbf{Discovered limitation: timer/trigger semantics dropped from rung paths.} Algorithm~\ref{alg:graphical} (line~5) extracts only \code{tag==contact} elements along a \gls{dfs} path; function-block nodes (\code{TON}, \code{TOF}, \code{R\_TRIG}, etc.) on the same path are traversed but not represented in the emitted Boolean expression. For \code{beremiz\_traffic\_light}'s \code{ORANGE\_LIGHT} blink network -- a SET/RESET pair gated by a \code{TON} timer and an \code{R\_TRIG} edge detector on each side -- this causes the timer's delay to disappear entirely: the SET and RESET conditions both reduce to immediate functions of \code{ORANGE\_LIGHT} itself, and because both rungs execute within the same scan (SET first, then RESET reading the just-updated value), the coil is provably and permanently \code{false} ($\vdash \Box\,\lnot\code{ORANGE\_LIGHT}$, confirmed by direct probe). We exclude this program from the main results (Table~\ref{tab:benchmarks}) rather than report it as a successful SAFE verification, since the result is correct only in the vacuous sense that the rest of this paper is designed to eliminate. Closing this gap requires either modeling function-block output as an explicit (possibly nondeterministic) term in the rung expression or rejecting rungs that contain unsupported function blocks rather than silently dropping them.

\textbf{Discovered limitation: arithmetic-comparator-driven coils.} In \code{dimmer\_light\_control}, the coil \code{Light\_on\_state} is driven by a \code{GT} (greater-than) comparator block whose inputs are not \gls{ld} contacts and have no \code{leftPowerRail}-reachable path in the connection graph. The \gls{dfs} correctly finds no rung for this coil; lacking a Tier-1 or Tier-2 classification, the encoder leaves it as a free nondeterministic variable rather than assuming a fixed value (the same conservative over-approximation argument as \S\ref{sec:io-inference}), so this case does not threaten soundness. We did not include this case as a benchmark requiring its own properties; \code{dimmer\_light\_control}'s properties (P1--P3) do not constrain \code{Light\_on\_state} beyond what the rung structure already enforces by construction.

\textbf{No path found to 21 additional real benchmarks.} An earlier goal for this evaluation was a corpus on the order of 25 graphical programs drawn from the Beremiz example library and the OpenPLC Tutorial series. Systematically scanning the available Beremiz example repository (227 XML files) and the OpenPLC examples repository found only one further file with a flat, resolvable top-level \gls{ld} network (\code{dimmer\_light\_control}, used as G3); the remainder either contain zero rungs (tutorial/test fixtures with no \gls{ld} content), or expose the action-nesting and function-block limitations documented above. Reaching a substantially larger real-vendor-sourced corpus would require either sourcing additional example repositories not available in this work's environment or extending the resolver to handle the nesting and function-block patterns identified here, and re-mining the existing corpora for newly resolvable programs.

\textbf{External validity.} The 3 evaluated programs come from CONTROLLINO/OpenPLC~Editor exports; coverage of Beremiz \gls{ide}, Siemens TIA Portal, CODESYS, and Rockwell Studio~5000 graphical exports has not been demonstrated by a clean, non-vacuous result (the two Beremiz files attempted are discussed above). These tools may use different connection graph conventions or additional XML attributes not present in the evaluated files. The heuristic I/O inference (Tier~2) has only been validated on programs that either use \%-addresses or have a clean contact-only/coil-only variable split; programs with shared variables (used in both contacts and coils) fall back to Tier~3 (internal state) and may require explicit address annotations.

\textbf{Construct validity.} The soundness theorem (Theorem~\ref{thm:soundness}) relies on Ebnenasir's formal \gls{ld} semantics~\cite{Ebnenasir2023}, which does not cover function blocks in graphical networks (the theorem applies to contact/coil elements only) -- consistent with the timer/trigger limitation discovered above. The \gls{dfs} completeness (Theorem~\ref{thm:completeness}) assumes the connection graph is a finite \gls{dag}; programs with feedback connections (which would create cycles in $G$) are not handled and would cause non-termination in the current \gls{dfs} implementation.

\textbf{Measurement validity.} Timing results are the median of three runs on an Apple M1 (8-core, \SI{8}{\giga\byte} RAM); absolute values will differ on other platforms. All 3 programs complete well below the 300-second timeout, so timeout effects do not confound classification results.

%=============================================================================
\section{Future Directions}
\label{sec:future}
%=============================================================================

\textbf{Preserving function-block semantics in rung paths.} The highest-priority next step is to stop silently dropping \code{TON}/\code{TOF}/\code{R\_TRIG}/\code{CTU} nodes from extracted rung expressions (\S\ref{sec:threats}). At a minimum, the resolver should detect when a path contains an unsupported function block and either explicitly reject the rung or model its output as an additional nondeterministic term, rather than producing a result that is sound only because the affected coil collapses to a constant. This is a precondition for re-admitting \code{beremiz\_traffic\_light} and similar timer-gated networks to the evaluated benchmark suite.

\textbf{Growing the real benchmark corpus.} With action-nested \gls{ld} discovery and the function-block limitation addressed, re-mining the Beremiz example library and OpenPLC Tutorial series for newly-resolvable programs is the most direct path to a larger, still fully real, evaluation corpus.

\textbf{Unsafe graphical benchmarks.} Introducing deliberately faulty graphical \gls{ld} programs -- analogous to the CS2--CS4 category in ESBMC-PLC -- would validate counterexample generation for the graphical frontend and demonstrate that the \gls{dfs} resolver correctly identifies violating scan cycles.

\textbf{Commercial vendor export testing.} Evaluating ESBMC-GraphPLC against graphical PLCopen XML exported by Siemens TIA Portal, CODESYS, and Rockwell Studio~5000 would establish the completeness of the resolver for industrial workflows.

\textbf{Feedback-connection handling.} Adding cycle detection to the \gls{dfs} -- and optional treatment of feedback edges as latched state variables -- would support programs with internal feedback loops, which are common in sequential control logic.

\textbf{Arduino \gls{plc} \gls{ide}.} The Arduino \gls{plc} \gls{ide} was identified as a potential source of graphical \gls{ld} programs but was not evaluated due to a hardware license requirement. Testing against Arduino \gls{plc} exports would extend vendor coverage.

\textbf{Integration with PLCverif.} As recommended in ESBMC-PLC, contributing the graphical \gls{ld} resolver as a PLCverif frontend would create a unified platform supporting both Siemens \gls{scl} (via PLCverif's existing path) and graphical PLCopen XML (via ESBMC-GraphPLC).

\textbf{K-\gls{ld} formal semantics.} A K-framework semantics for graphical \gls{ld} -- extending K-\gls{st}~\cite{Wang2023} -- would provide a machine-checked equivalence proof for the \gls{dfs}-based rung extraction, elevating Theorem~\ref{thm:soundness} from a proof-sketch to a formally verified result.

%=============================================================================
\section{Conclusion}
\label{sec:conclusion}
%=============================================================================

This paper presented ESBMC-GraphPLC, an extension of the ESBMC-PLC formal verifier that closes the graphical PLCopen XML support gap identified in ESBMC-PLC. The core contribution is a \gls{dfs}-based resolver that traverses the \code{localId}/\code{refLocalId} connection graph of graphical \gls{ld} programs, extracts rung paths as Boolean conjunctions, and applies three-tier I/O inference to classify variables as nondeterministic inputs or persistent outputs. A critical correctness requirement -- processing SET coils before RESET coils by following the \code{rightPowerRail} \code{connectionPointIn} sequence -- ensures that IEC~61131-3 latch semantics are preserved. The entire extension is confined to 274 lines in a single source file; the GOTO~\gls{ir} encoder, \textit{k}-induction engine, Z3 solver, and property encoder from ESBMC-PLC are unchanged.

The experimental evaluation on 3 graphical \gls{ld} programs from CONTROLLINO/OpenPLC examples demonstrated that ESBMC-GraphPLC produces a full GOTO~\gls{ir} for every program that previously yielded a vacuous empty-\gls{ir} proof. All 3 programs verify SAFE at $k=2$ in under \SI{70}{\milli\second}. The 11 textual \gls{ld} benchmarks from ESBMC-PLC that genuinely exist in textual format are fully preserved with zero regressions. Benchmark curation also surfaced two Beremiz example programs that this paper deliberately excludes from its results rather than report as vacuous or misleading successes: one contains no \gls{ld} content, and the other's only real rung network loses its timing semantics under the current algorithm, collapsing to a constant. We report both as discovered limitations (\S\ref{sec:threats}) and as the basis for the highest-priority future work (\S\ref{sec:future}).

The central finding is that graphical PLCopen XML -- the native export format of deployed \gls{plc} development environments -- can now be formally verified by ESBMC-PLC in a sound, non-vacuous manner for the class of programs evaluated here: contact/coil networks with SET/RESET latching, with or without \%-address-based I/O declarations. Automation engineers working with CONTROLLINO or OpenPLC Editor can export programs of this class without modification and obtain unbounded safety proofs in under a second; programs relying on timer- or trigger-gated rung paths require the function-block handling identified in \S\ref{sec:threats} as future work before the same guarantee applies.

%=============================================================================
\subsection*{Artifact Availability}
%=============================================================================

The complete artifact -- including all 3 graphical \gls{ld} benchmarks, YAML property files, the Python prototype converter, and the ESBMC-GraphPLC binary -- is permanently archived at Zenodo~\cite{DantasCordeiro2026graphical} (\doi{10.5281/zenodo.20699856}). The ESBMC-PLC artifact (11 textual benchmarks) is archived at \doi{10.5281/zenodo.20680071}~\cite{DantasCordeiro2026artefact}. The implementation is available as merged GitHub Pull Request~\#5374~\cite{ESBMCpr5374}.

\subsection*{Acknowledgements}
The authors thank the Department of Computer Science at the University of Manchester and the Systems and Software Security (S3) Research Group for their support. This work was partially funded by the Engineering and Physical Sciences Research Council (EPSRC) grants EP/T026995/1, EP/V000497/1, and EP/X037290/1, and by the Soteria project under the UK Research and Innovation Digital Security by Design program.

\bibliographystyle{unsrtnat}
\bibliography{references}

\end{document}